\documentclass[sigconf]{acmart}
\AtBeginDocument{%
  }
\usepackage{newunicodechar}
\newunicodechar{ }{\,}
\usepackage[utf8]{inputenc} % allow utf-8 input
\usepackage[T1]{fontenc}    % use 8-bit T1 fonts
\usepackage{hyperref}       % hyperlinks
\usepackage{url}            % simple URL typesetting
\usepackage{booktabs}       % professional-quality tables
\usepackage{amsfonts}       % blackboard math symbols
\usepackage{nicefrac}       % compact symbols for 1/2, etc.
\usepackage{microtype}      % microtypography
\usepackage{xcolor}         % colors
\usepackage{amsmath}
\usepackage{graphicx}
\usepackage{float}    % required for [H]
\copyrightyear{2026}
\acmYear{2026}
\setcopyright{cc}
\setcctype{by}
\acmConference[KDD '26]{Proceedings of the 32nd ACM SIGKDD Conference on Knowledge Discovery and Data Mining V.2}{August 09--13, 2026}{Jeju Island, Republic of Korea}
\acmBooktitle{Proceedings of the 32nd ACM SIGKDD Conference on Knowledge Discovery and Data Mining V.2 (KDD '26), August 09--13, 2026, Jeju Island, Republic of Korea}
\acmDOI{10.1145/3770855.3818872}
\acmISBN{979-8-4007-2259-2/2026/08}

\begin{document}
\title{How Far Can You Grow? Characterizing the Extrapolation Frontier of Graph Generative Models for Materials Science}
\author{Can Polat}
\email{can.polat@tamu.edu}
\affiliation{%
  \institution{Texas A\&M University}
  \city{College Station}
  \state{Texas}
  \country{USA}
}

\author{Erchin Serpedin}
\email{eserpedin@tamu.edu}
\affiliation{%
  \institution{Texas A\&M University}
  \city{College Station}
  \state{Texas}
  \country{USA}
}

\author{Mustafa Kurban}
\authornote{Corresponding authors.}
\email{kurbanm@ankara.edu.tr}
\affiliation{%
  \institution{Ankara University}
  \city{Ankara}
  \country{Turkey}
}
\affiliation{%
  \institution{Texas A\&M University at Qatar}
  \city{Doha}
  \country{Qatar}
}

\author{Hasan Kurban}
\authornotemark[1]
\email{hkurban@hbku.edu.qa}
\affiliation{%
  \institution{College of Science \& Engineering\\ Hamad Bin Khalifa University}
  \city{Doha}
  \country{Qatar}
}

\renewcommand{\shortauthors}{Polat et al.}

\begin{teaserfigure}
  \includegraphics[width=\textwidth]{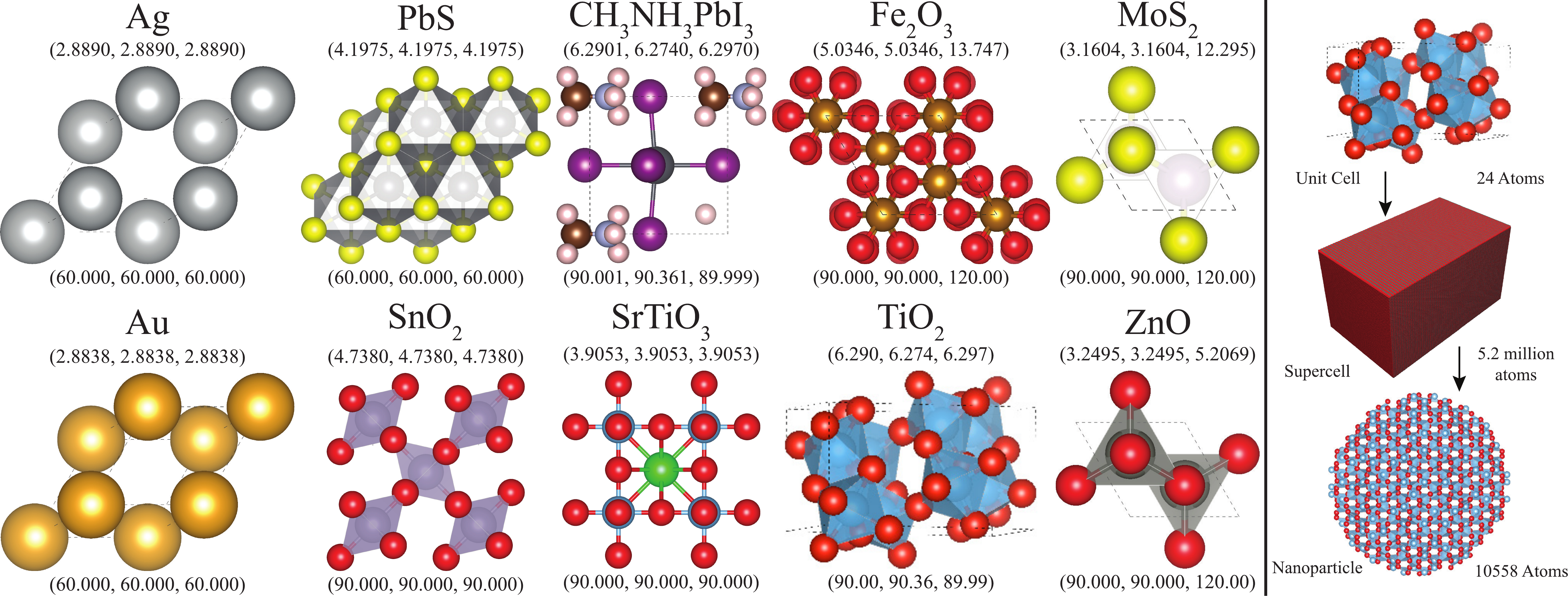}
  \Description{Teaser overview}
  \caption{Detailed overview of RADII. (left) Unit cells of the materials (Ag, Au, CH$_3$NH$_3$PbI$_3$, Fe$_2$O$_3$, MoS$_2$, PbS, SnO$_2$, SrTiO$_3$, TiO$_2$, and ZnO) arranged left to right. The lattice constants a, b and c denote the cell edge lengths, while the angles $\alpha$, $\beta$ and $\gamma$ specify the inter-edge angles between b–c, a–c and a–b, respectively. (right) Workflow for generating radius-resolved nanoparticles.}
  \label{fig:datasetProcess}
\end{teaserfigure}
\begin{abstract}
Every generative model for crystalline materials harbors a critical structure size beyond which its outputs quietly become unreliable; we call this the \textit{extrapolation frontier}. Despite its direct consequences for nanomaterial design, this frontier has never been systematically measured. We introduce RADII, a radius-resolved benchmark of ${\sim}$75,000 crystal-derived nanoparticle structures (33--11,298 atoms) that treats radius as a continuous scaling knob to trace generation quality from in-distribution to out-of-distribution regimes under leakage-free splits. Each model is conditioned on the target composition and atom count, isolating geometric extrapolation as the evaluation variable. RADII provides frontier-specific diagnostics: per-radius error profiles pinpoint each architecture's scaling ceiling, surface--interior decomposition tests whether failures originate at boundaries or in bulk, and cross-metric failure sequencing reveals which aspect of structural fidelity breaks first. Benchmarking five state-of-the-art architectures, we find that: (i) well-behaved models degrade by ${\sim}13\%$ in global positional error beyond training radii, while divergent models exhibit poor absolute fidelity across scales; local bond fidelity varies sharply across architectures, from negligible degradation to more than $2\times$ error growth;
(ii) no two architectures share the same failure sequence, revealing the frontier as a multi-dimensional surface shaped by model family; and (iii) well-behaved models conform to the expected geometric scaling exponent $\alpha \approx 1/3$ whose in-distribution fit accurately predicts out-of-distribution error, making their frontiers quantitatively forecastable. Scaling MatterGen to its published parameter count stabilizes sampling but does not close the extrapolation frontier, while DiffCSP remains unstable even at published scale. These findings establish output scale as a first-class evaluation axis for geometric generative models. The dataset, generation pipeline, and implementations are available at \url{https://github.com/KurbanIntelligenceLab/RADII}.
\end{abstract}


\maketitle

\section{Introduction}\label{sec:intro}
Generative models for crystalline materials are routinely evaluated at the scales on which they were trained, creating an illusion of reliability that shatters once output size departs from the training distribution. We refer to the critical size threshold at which this collapse occurs as the \textit{extrapolation frontier}: a quantity that, despite its direct implications for nanomaterial design, has never been systematically measured. Nanostructured materials drive applications from photovoltaics to chemical sensing \cite{simonov2020designing}, with properties governed both by the periodic symmetry of primitive unit cells and by the finite morphologies of nanoparticles \cite{cao2022tio2}. These regimes are conventionally treated separately: crystal simulations propagate ideal lattices, whereas nanoparticle workflows construct finite clusters through empirical or first-principles refinement \cite{levi1997theory}. Both approaches face steep scalability barriers across compositions, sizes, and orientations \cite{surek2005crystal}. Density functional theory (DFT) \cite{orio2009density} and its tight-binding approximation (DFTB) \cite{elstner2014density} deliver accurate energetics, but DFT's cubic scaling restricts exploration to modest system sizes \cite{cohen2008insights}; DFTB reduces this cost yet remains limited for large-scale nanostructure generation \cite{liu2019efficient, qi2013comparison}.

Machine learning offers a scalable alternative \cite{karniadakis2021physics}. Graph neural networks \cite{schnet} achieve strong predictive performance on molecular and crystalline benchmarks \cite{md17, chanussot2021open}. Equivariant architectures \cite{satorras2021n, fuchs2020se, spherenet} improve data efficiency through geometric symmetry constraints, and multimodal methods \cite{pure2dope, rollins2024molprop} incorporate diverse input modalities. Generative models, including symmetry-preserving diffusion and flow approaches \cite{ levy2025symmcd, kelvinius2025wyckoffdiff}, have advanced crystal structure generation considerably. However, existing evaluations focus predominantly on small molecules or bulk periodic crystals at roughly fixed output sizes. This evaluation paradigm conceals a fundamental failure mode: it remains unknown \textit{how far} beyond their training distribution these models can reliably generate structures, and \textit{how} generation quality degrades as output size increases.

RADII addresses this gap by systematically mapping the extrapolation frontier across architectures, materials, and metrics. Using radius as a continuous scaling knob, the benchmark links primitive unit cells extracted from published crystallographic data to nanoparticles across 25 size configurations spanning 6--30~\AA, yielding approximately 75,000 structures containing 33--11,298 atoms. A leakage-free data split cleanly separates orientation interpolation (in-distribution, ID) from radius extrapolation (out-of-distribution, OOD), enabling precise identification of each model's scaling ceiling. Each generation task is conditioned on the target composition and atom count, so RADII isolates geometric extrapolation rather than conflating it with composition prediction or assignment ambiguity. For every model--material pair, RADII traces per-radius error profiles to locate the frontier, decomposes failures into surface versus interior contributions, and quantifies extrapolation gap severity across complementary metrics. Benchmarking state-of-the-art generative models reveals that all architectures degrade sharply, but the frontier location and failure signature vary systematically with material symmetry and architecture family. These findings establish output scale as a first-class evaluation axis for geometric generative models and demonstrate that current scaling limits are predictable rather than random. By providing a reproducible, geometry-grounded diagnostic testbed, RADII lays the foundation for developing architectures that generalize beyond their training horizon.

\section{Related Work}\label{sec:related}

\subsection{Geometric Graph Generation for Materials}

Generative modeling of atomic structures has advanced rapidly, yet nearly all progress has been measured at fixed output scales. Crystalline unit cells encode the symmetry operations and atomic motifs that generative models must reproduce \cite{kittel}, and when bulk crystals are truncated into finite clusters, their morphologies follow orientation-dependent surface energies governed by the Gibbs--Wulff theorem \cite{li2016gibbs, wuffCrystal, wuff2}. These geometric effects produce controlled, symmetry-consistent deviations from ideal bulk structure as system size varies, precisely the kind of structured distribution shift needed to characterize extrapolation frontiers. Quantum-confinement phenomena \cite{bera2010quantum} lie outside scope; RADII retains only the geometric variation needed to probe whether generation quality holds as output size departs from training conditions. By employing deterministic, symmetry-preserving construction rather than simulation-driven morphologies, the benchmark isolates scale as an independent variable for evaluating generative models \cite{yang2022big}.

Graph neural networks \cite{schnet} and equivariant architectures \cite{satorras2021n, fuchs2020se, spherenet} have achieved strong performance on molecular and crystalline benchmarks \cite{md17, chanussot2021open}, with multimodal methods \cite{pure2dope, rollins2024molprop} further expanding input representations. Generative models, including diffusion-based approaches \cite{graphdf, jiao2024space}, flow-matching methods \cite{levy2025symmcd}, and symmetry-aware generators \cite{kelvinius2025wyckoffdiff}, have advanced crystal structure generation considerably. Recent symmetry-aware crystal structure prediction methods such as EquiCSP~\cite{equicsp}  and SGEquiDiff~\cite{sgequidiff} further exploit space-group equivariance and Wyckoff-position priors to improve periodic structure generation. However, these periodic inductive biases assume translational symmetry that is explicitly broken in finite nanoparticles, and RADII's results empirically confirm that architectures designed under such assumptions can fail to scale to non-periodic clusters. More broadly, all these models are developed and validated on datasets where output size is approximately constant (e.g., QM9 at ${\sim}9$ atoms, MP-20 at ${\sim}20$ atoms per cell), meaning that their behavior under size extrapolation remains entirely uncharacterized. Unified geometric representation benchmarks such as Geom3D~\cite{geom3d} have broadened evaluation across tasks and representations but likewise do not vary output scale. RADII is designed to fill exactly this diagnostic gap.

\subsection{Scalability Limits of Physics-Based Methods}

The extrapolation frontier matters in practice because physics-based alternatives cannot cover the size ranges that generative models are increasingly asked to target. Kohn--Sham DFT provides accurate energetics but scales as $\mathcal{O}(N^3)$, restricting routine simulations \cite{bickelhaupt2000kohn, yu2016perspective}. Linear-scaling approaches such as ONETEP \cite{baer1997sparsity, skylaris2005introducing}, semi-empirical techniques like DFTB \cite{zheng2005performance, kim2019multiscale}, and classical or ML-based interatomic potentials \cite{daw1984embedded, mahata2022modified} extend this ceiling but cannot reliably generate structures approaching RADII's 11{,}298-atom upper bound. This computational bottleneck is precisely why ML-based generation is attractive for nanomaterial design, and why understanding where these models break under size extrapolation is urgent. Because RADII targets geometric scaling behavior rather than energetic accuracy, physics-based generation is both unnecessary and infeasible at benchmark scale; instead, deterministic symmetry-preserving construction provides the scalable ground truth needed to map extrapolation frontiers \cite{kurban2024enhancing}.

\subsection{Evaluation Gaps in Existing Benchmarks}
Benchmark datasets have driven geometric deep learning forward, yet none systematically probe size extrapolation. QM7 \cite{blum, rupp}, MD22 \cite{chmiela2023accurate}, PubChemQC \cite{kim2025pubchem}, NablaDFT \cite{khrabrov2022nabladft}, and QH9 \cite{yu2024qh9} target molecular energetics or Hamiltonians; Perov-5 \cite{perov5a, perov5b} and Carbon-24 \cite{pickard2020airss} catalog crystalline frameworks; MatBench \cite{dunn2020benchmarking}, OC20 \cite{chanussot2021open}, OC22 \cite{tran2023open}, LAMBench \cite{lambench}, and CrysMTM \cite{polat2025crysmtm} benchmark property or force predictions on surface and bulk structures. All evaluate at approximately fixed output scale; none treats output size as a continuous evaluation axis, so extrapolation frontiers have never been measured. The broader OOD generalization literature for GNNs, including causality-based and augmentation-based approaches, underscores the importance of controlled distribution shifts; RADII provides exactly such a structured geometric shift and may serve as a complementary testbed for graph OOD methods. Power-law scaling relationships between error and system size are well-studied in neural scaling laws~\cite{kaplan2020scaling} and finite-size scaling in statistical physics~\cite{fisher1972scaling}, but have not been applied to characterize geometric generative models; the $\alpha \approx 1/3$ exponent identified in Section~\ref{subsec:exp_scaling} connects to these traditions by relating generation error to the linear dimension of the structure. Closest to our setting, C2NP~\cite{c2np} formalizes a unit-cell-to-nanoparticle task with spherical truncation and rotation-stratified splits, evaluating multiple architectures and including a reverse (nanoparticle$\,\to\,$unit cell) direction. RADII builds on this foundation with two distinct extensions: (i) frontier-specific diagnostics not present in C2NP (surface--interior decomposition, coordination correlation, cross-metric failure sequencing, and degradation ratios), and (ii) explicit scaling-law fits with OOD residual analysis that make frontiers quantitatively forecastable.

\begin{figure*}[ht]
    \centering   \includegraphics[width=0.9\linewidth] {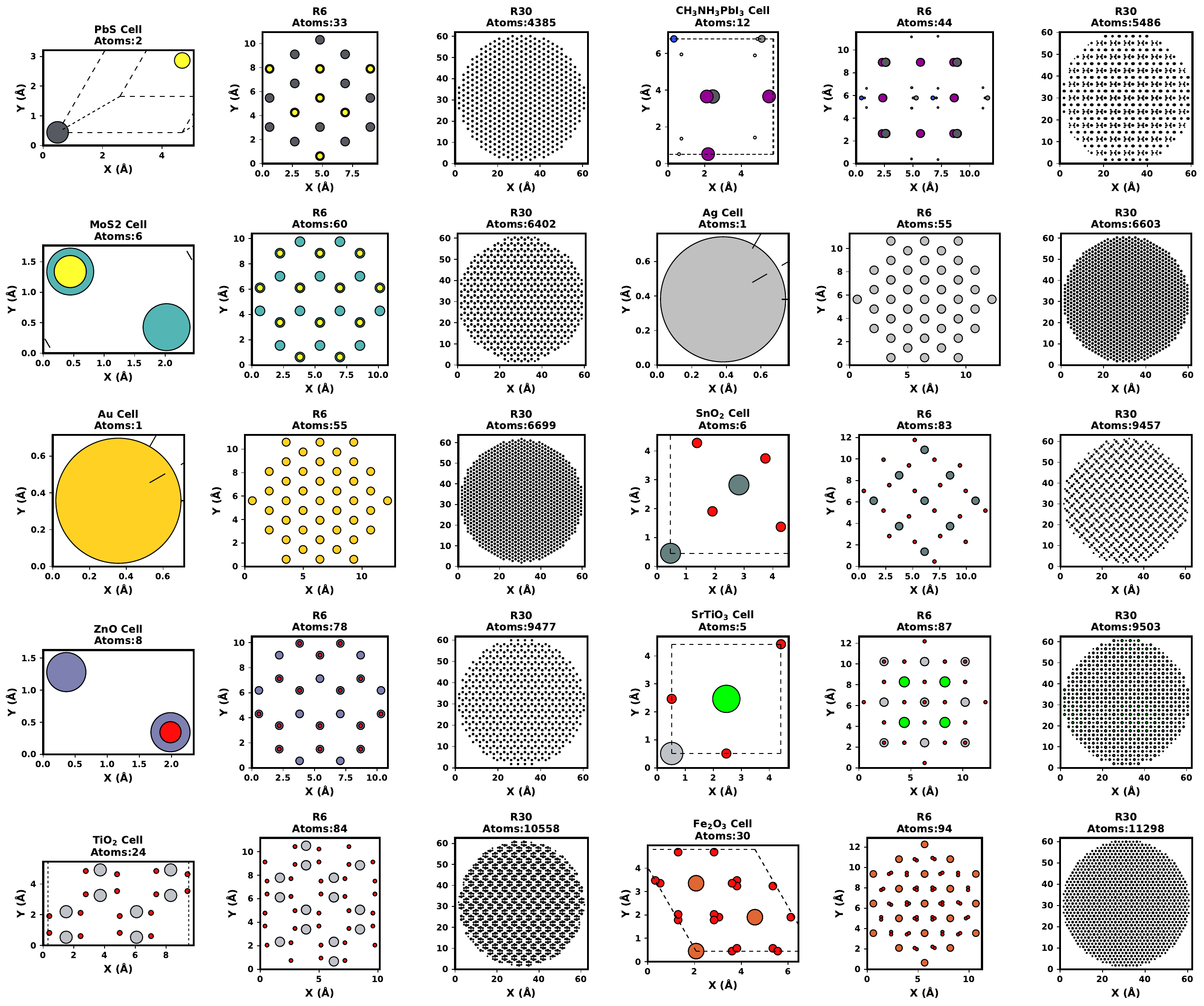}
    \Description{Dataset Overview}
    \caption{From primitive cell to radius‑controlled nanoclusters. For each material in the dataset, the panels show, left to right, the primitive unit cell followed by its canonical $R = 6$~\AA\ and $R = 30$~\AA nanoparticles. Materials are arranged from top to bottom in ascending order of the atom count in their $R{=}30$ cluster, illustrating how coordination environments and bulk‑like cores emerge with increasing radius. All views share a common Ångström scale. Atom colours follow the conventional CPK palette.}
    \label{fig:datasetOverall}
\end{figure*}

\section{RADII Construction}\label{sec:dataset}

\subsection{Task Formulation}
\label{subsec:task}

Given a primitive unit cell $\mathcal{U}_m = (\{\mathbf{b}_j, z_j\}_{j=1}^M, \mathbf{v}_1, \mathbf{v}_2, \mathbf{v}_3)$ encoding basis positions, species, and lattice vectors, a target radius $R$, and the corresponding ground-truth species sequence $\mathbf{z} = (z_1,\dots,z_{N(R)})$, the model generates a finite set of 3D atomic positions:
\begin{equation}
f : (\mathcal{U}_m,\; R,\; \mathbf{z}) \;\longmapsto\; \{\hat{\mathbf{x}}_i\}_{i=1}^{N(R)} \subset \mathbb{R}^3,
\label{eq:task}
\end{equation}
where $\hat{\mathbf{x}}_i \in \mathbb{R}^3$ is the predicted position of atom $i$, the atom count $N(R) = |\mathcal{S}(R)|$ is determined by $R$ via the spherical truncation in Eq.~\eqref{eq:structure}, and atomic species are not predicted but supplied as inputs in $\mathbf{z}$. This conditioning establishes a natural one-to-one correspondence between predicted and reference atoms without requiring Hungarian matching or any other assignment algorithm. We emphasize that this formulation evaluates a \emph{geometry-only} subtask of generative modeling: chemical identities are handled deterministically, and each model receives the exact stoichiometry and species ordering of the target nanoparticle. This design cleanly isolates geometric extrapolation as the variable of interest; extending the benchmark to an unconditioned track where models must additionally predict composition and ordering is a natural direction discussed in Section~\ref{sec:limitation}. Evaluating~\eqref{eq:task} across all 25 radii traces each model's quality from within the training distribution to beyond the extrapolation frontier.

\paragraph*{Why no deterministic reconstruction baseline.}
Since the ground-truth nanoparticle \emph{is} the deterministic spherical truncation of the unit cell (Eq.~\ref{eq:structure}), a rule-based baseline that reconstructs the reference from the input lattice parameters would trivially achieve zero error: the lattice vectors and basis positions fully specify every atom's position. The task evaluates \emph{generation} (producing atomic coordinates from learned representations), not reconstruction from explicit lattice parameters. A supervised coordinate-regression baseline (e.g., an MLP mapping unit-cell features to atom positions) faces the same triviality: with the unit cell, radius, and atom ordering all provided, the mapping is deterministic and learnable to near-zero error given sufficient capacity, providing no diagnostic value for measuring extrapolation. Instead, the ID performance of each generative model serves as its own architecture-specific reference against which OOD degradation is measured. The benchmark question is precisely whether a learned generative model can preserve and extrapolate these structural regularities through its learned representations as output size grows, an $O(N)$ coordinate-placement problem that an explicit lattice rule, but not a learned model, resolves for free.

\subsection{Material Selection and Structure Generation}
\label{subsec:structgen}

The benchmark spans FCC metals (Ag~\cite{silver2002crc}, Au~\cite{gold2002crc}), corundum Fe$_2$O$_3$~\cite{finger1980fe2o3}, a layered transition-metal dichalcogenide MoS$_2$~\cite{wyckoff1963mos2, grau2002mos2}, rock-salt PbS~\cite{wyckoff1963pbs}, rutile-type oxides SnO$_2$~\cite{baur1971sno2} and TiO$_2$~\cite{horn1972tio2}, perovskites SrTiO$_3$~\cite{mitchell2000srtio3} and CH$_3$NH$_3$PbI$_3$~\cite{walsh2019hybrid}, and wurtzite ZnO~\cite{wyckoff1963zno}. RADII is built entirely from experimentally reported crystal structures, with the ten primitive unit cells extracted from published CIFs and used as input conditioning for all models. Deterministic spherical truncation produces idealized crystal-derived nanoparticle references containing 33--11{,}298 atoms; after orientation sampling the benchmark comprises ${\sim}$75{,}000 structures, partitioned into 48{,}000 training, 13{,}500 ID-test, and 13{,}480 OOD-test structures. Figure~\ref{fig:datasetOverall} summarizes benchmark structure and key statistics, and Appendix~\ref{app:dataset} reports per-material and per-radius statistics.

\paragraph*{Supercell construction.}
Let $N_{\mathrm{rep}}=60$. The periodic lattice of atomic sites is
\begin{equation}
\begin{aligned}
\mathcal{L}
=\Big\{\, &
n_1\mathbf{v}_1 + n_2\mathbf{v}_2 + n_3\mathbf{v}_3 + \mathbf{b}_j \\
&\Big|\;
n_1,n_2,n_3 \in \{0,\dots,N_{\mathrm{rep}}-1\},\;
j \in \{1,\dots,M\}
\Big\},
\end{aligned}
\label{eq:lattice}
\end{equation}
where $\mathbf{v}_1,\mathbf{v}_2,\mathbf{v}_3\in\mathbb{R}^3$ are the primitive lattice vectors and $\{\mathbf{b}_j\}_{j=1}^M$ are the basis positions.

\paragraph*{Scale-resolved nanoparticle.}
Selecting a central reference site $\mathbf{b}_0$ as the image of $\mathbf{b}_1$ in the central unit cell, the nanoparticle at radius $R$ is
\begin{equation}
\mathcal{S}(R) = \{\mathbf{x}\in\mathcal{L}\mid \|\mathbf{x}-\mathbf{b}_0\|_2\le R\},
\label{eq:structure}
\end{equation}
i.e., all lattice sites within a sphere of radius $R$. We sample $K=25$ radii uniformly:
\begin{equation}
R_k = R_{\min} + (k-1)\Delta R, \qquad k=1,\dots,25,
\end{equation}
with $R_{\min}=6$~\AA, $R_{\max}=30$~\AA, and $\Delta R=1$~\AA, yielding structures of 33--11,298 atoms.

\subsection{Radius Split Protocol}
\label{subsec:splits}

The 25 radii ($R \in \mathcal{R} = \{6, 7, \dots, 30\}$~\AA) are partitioned into three disjoint groups to cleanly separate interpolation from extrapolation. ID test radii $\mathcal{R}_{\mathrm{ID}} = \{11, 13, 15, 17, 19, 21\}$~\AA\ are held out from the training set but interleaved within the training radius range of 8--28~\AA, and are evaluated under unseen orientations to measure interpolation quality.

% ID test radii $\mathcal{R}_{\mathrm{ID}} = \{11, 13, 15, 17, 19, 21\}$~\AA\ are interleaved within train range of 8 to 28 but held out, and evaluated under unseen orientations to measure interpolation quality.
The four OOD test radii $\mathcal{R}_{\mathrm{OOD}} = \{6, 7, 29, 30\}$~\AA\ lie strictly outside the training range, below (6, 7~\AA) and above (29, 30~\AA), probing extrapolation in both the small-particle and large-particle regimes. Because atom count scales as $R^3$, these radius offsets correspond to large shifts in structural complexity: the training radii span 81--9{,}148 atoms, whereas the OOD radii span 33--11{,}298 atoms, about 59\% smaller than the smallest and 24\% larger than the largest training structure. The leakage-free guarantee is two-fold: (i)~no OOD or ID radius appears during training, and (ii)~ID and OOD test orientations are excluded from training orientations via the angular exclusion constraint described below.

\subsection{Quaternion-Based Orientation Sampling}
\label{subsec:rotations}

For each nanoparticle $\mathcal{S}(R)$, we generate rigidly rotated copies using unit quaternions to provide diverse input orientations for training and evaluation.

\paragraph*{Quaternion distance.}
Rotation is represented by $q=(q_x,q_y,q_z,q_w)$ $\in\mathbb{S}^3$ with canonical sign $q_w\ge 0$. The geodesic distance is
\begin{equation}
d(q_1,q_2) = 2\arccos\bigl(|\langle q_1,q_2\rangle|\bigr) \in [0,\pi].
\label{eq:quat_geodesic}
\end{equation}

\paragraph*{Greedy angular separation.}
Given spacing $\Delta\theta > 0$, we construct a set $\mathcal{Q}$ by rejection sampling from the Haar-uniform distribution on $\mathrm{SO}(3)$, accepting candidates sequentially subject to
\begin{equation}
\max_{q'\in\mathcal{Q}} |\langle q',q\rangle| \le \cos(\Delta\theta/2),
\label{eq:greedy_sep}
\end{equation}
guaranteeing $d(q_i,q_j) \ge \Delta\theta$ for all accepted pairs (Proposition~\ref{prop:angular_separation}).

\begin{proposition}[Angular separation guarantee]
\label{prop:angular_separation}
For any $q_i,q_j\in\mathcal{Q}$ returned by the greedy procedure at spacing $\Delta\theta$,
\begin{equation}
|\langle q_i,q_j\rangle| \le \cos(\Delta\theta/2) \;\;\Longleftrightarrow\;\; d(q_i,q_j) \ge \Delta\theta.
\end{equation}
\end{proposition}
\noindent\textit{Proof.} Follows directly from $d(q_i,q_j)=2\arccos(|\langle q_i,q_j\rangle|)$ and the greedy acceptance rule. \hfill$\square$

\paragraph*{Split-specific construction and design rationale.}
RADII uses split-dependent angular spacings: $\Delta\theta_{\mathrm{train}}=16^\circ$, $\Delta\theta_{\mathrm{ID}}=14^\circ$, $\Delta\theta_{\mathrm{OOD}}=12^\circ$. The spacings decrease from train to test to provide progressively denser angular coverage for evaluation splits, ensuring thorough probing of orientation sensitivity at test time. The generation proceeds in two passes to enforce strict leakage prevention:

\textit{Pass~1 (Train + ID).} A single global quaternion grid $\mathcal{Q}_{\mathrm{train}}$ is generated once at spacing $\Delta\theta_{\mathrm{train}}$; per-structure subsets are drawn via deterministic seeding to ensure reproducibility. ID quaternions are then sampled subject to an exclusion constraint against $\mathcal{Q}_{\mathrm{train}}$:
\begin{equation}
\max_{q'\in\mathcal{Q}_{\mathrm{train}}} |\langle q',q\rangle| \le \cos(\delta_{\mathrm{ID}}/2),
\label{eq:exclusion_id}
\end{equation}
with exclusion margin $\delta_{\mathrm{ID}}=8^\circ$.

\textit{Pass~2 (OOD).} OOD quaternions are sampled subject to exclusion against the union $\mathcal{Q}_{\mathrm{train}} \cup \mathcal{Q}_{\mathrm{ID}}$ of all previously generated orientations:
\begin{equation}
\max_{q'\in\mathcal{Q}_{\mathrm{train}} \cup \mathcal{Q}_{\mathrm{ID}}} |\langle q',q\rangle| \le \cos(\delta_{\mathrm{OOD}}/2),
\label{eq:exclusion_ood}
\end{equation}
with stricter exclusion margin $\delta_{\mathrm{OOD}}=10^\circ$ to provide a wider angular buffer against training orientations, reflecting the stronger isolation required for out-of-distribution evaluation.

Fixed left-multiplication offsets ($q_{\mathrm{ID}}=\mathrm{Euler}_{xyz}(20^\circ,30^\circ,45^\circ)$, $q_{\mathrm{OOD}}=\mathrm{Euler}_{xyz}(50^\circ,70^\circ,90^\circ)$) shift each split's quaternion candidates into a distinct region of $\mathrm{SO}(3)$ prior to exclusion checking. These offsets are chosen to be well-separated from each other and from the identity (which seeds the training grid), ensuring that each split explores a geometrically distinct portion of orientation space. The two-pass architecture guarantees that OOD orientations are excluded against both training and ID orientations, providing a strictly stronger leakage-free guarantee than independent sampling.

Highly symmetric structures that map multiple rotations to identical coordinates are deduplicated by hashing rounded coordinates at tolerance $\varepsilon=10^{-6}$, retaining only unique orientations. After deduplication, the final dataset contains 74,980 structures in total, partitioned as 48,000 training structures, 13,500 ID test structures, and 13,480 OOD test structures.

\subsection{Evaluation Metrics}
\label{subsec:metrics}

RADII's metrics answer a central question: \textit{at what size does generation quality collapse, and what breaks first?} They are organized into three tiers: generation quality measures tracked per radius, failure decomposition diagnostics that localize errors within each structure, and frontier characterization metrics that quantify scaling ceilings.

\paragraph*{Correspondence guarantee.}
All metrics below rely on a one-to-one mapping between predicted and ground-truth atoms. As described in Section~\ref{subsec:task}, this correspondence is guaranteed by construction: the conditioning provides the exact atom count $N(R)$ and species sequence, and atom ordering is inherited from the input. Kabsch alignment therefore operates on paired, equal-sized, species-matched point sets without requiring permutation search.

\paragraph*{On assignment-free alternatives.}
Because the correspondence is guaranteed, assignment-free metrics such as Chamfer distance or earth mover's distance are not required for correctness. However, such metrics, along with radial distribution function (RDF) divergences, would provide complementary, correspondence-independent views of generation quality and could help cross-validate the failure sequences reported in Section~\ref{subsec:exp_sequence}. We discuss their inclusion in Section~\ref{sec:limitation}. The present metric suite is chosen to maximize diagnostic specificity: RMSD localizes global errors, BondMAE captures local chemical fidelity, and CoordCorr tracks topological preservation; aggregate distribution-level metrics would obscure these distinctions.

\subsubsection{Generation Quality Measures}
\label{subsubsec:quality}

After centering both prediction $P$ and ground truth $G$ and computing the optimal Kabsch rotation $\mathbf{R}^\star$, we define: $\mathrm{RMSD}(P,G)=\sqrt{\frac{1}{N}\|\tilde{P}-\tilde{G}\|_F^2},
$ where $\tilde{P}=\mathrm{center}(P)\mathbf{R}^\star$ and $\tilde{G}=\mathrm{center}(G)$. The one-to-one correspondence required by Kabsch alignment is guaranteed by the task conditioning: both $P$ and $G$ contain exactly $N(R)$ atoms with matching species, and ordering is preserved from the input.

\paragraph*{Local bond-length MAE}
Let $D_k(X)\in\mathbb{R}^{N\times k}$ be the $k$-nearest-neighbor distances via KD-tree. Flattening and sorting into vectors $d_P, d_G$:
\begin{equation}
\mathrm{BondMAE}_k(P,G)=\frac{1}{m}\sum_{i=1}^m |(d_P)_i-(d_G)_i|, \qquad m=\min(|d_P|,|d_G|).
\end{equation}
We note that this formulation compares globally sorted distance vectors rather than per-atom neighbor lists, which may conflate distinct local environments; per-atom kNN matching or bond-angle distributions could provide finer localization and are considered for future inclusion (Section~\ref{sec:limitation}). As a robustness check, a per-atom variant that compares each atom's neighbor list individually yields identical model rankings (Spearman $\rho = 1.0$; Appendix~\ref{app:sensitivity}), confirming that global sorting does not mislead the architectural comparison. Tracking BondMAE alongside RMSD nonetheless distinguishes models that lose local chemical order from those that maintain short-range structure but produce incorrect global morphology.

\subsubsection{Failure Decomposition Diagnostics}
\label{subsubsec:decomposition}

\paragraph*{Surface--interior error ratio.}
Let $S$ and $I$ index the outermost and innermost $25\%$ of atoms by distance from the centroid, where shell membership is defined on the \emph{ground-truth} structure and per-atom errors are computed from the Kabsch-aligned prediction using the guaranteed atom correspondence:
\begin{equation}
\mathrm{SurfIntRatio}=\frac{\mathrm{SurfRMSD}}{\mathrm{IntRMSD}+10^{-8}}.
\end{equation}
A ratio increasing with radius indicates boundary-driven collapse; a stable ratio signals uniform degradation. The $25\%$ shell fraction is a default; Appendix~\ref{app:sensitivity} shows the ID--OOD conclusions hold across fractions from $15\%$ to $30\%$.

\paragraph*{Coordination preservation.}
Using KD-tree ball queries with cutoff $r_c$ (default $r_c = 3.0$~\AA; cross-architecture rankings are stable for $r_c \in [3.0, 4.5]$~\AA, Appendix~\ref{app:sensitivity}), per-atom coordination numbers yield:
\begin{equation}
\mathrm{CoordCorr}(P,G)=\mathrm{corr}(c_P, c_G)\in[-1,1].
\end{equation}
A sharp drop at a specific radius signals the model has exceeded the scale at which it maintains local structural rules.

\subsubsection{Frontier Characterization}
\label{subsubsec:frontier}

These metrics operate on per-radius error profiles $m(R)$ rather than individual structures.

\paragraph*{ID--OOD degradation ratio.}
\begin{equation}
\mathrm{Degrad}(m)=\frac{\frac{1}{|\mathcal{R}_{\mathrm{OOD}}|}\sum_{R\in\mathcal{R}_{\mathrm{OOD}}} m(R)}{\frac{1}{|\mathcal{R}_{\mathrm{ID}}|}\sum_{R\in\mathcal{R}_{\mathrm{ID}}} m(R) + 10^{-8}}.
\end{equation}
Values near $1$ indicate robust scaling; values $\gg 1$ indicate failure to generalize beyond training sizes.

\paragraph*{Frontier radius.}
For quality threshold $\tau$: $r^\star(m, \tau) = \max\{R \in \mathcal{R} : m(R) \leq \tau\}.$ Comparing $r^\star$ across models, materials, and metrics provides a compact summary of each architecture's scaling ceiling.

\textbf{Reproducibility.} All construction is fully deterministic: a global seed with per-structure FNV-1a hashing of (material, radius, split) ensures identical outputs across runs. The released repository includes CIF-to-nanoparticle scripts, split configurations, quaternion generation code, and all evaluation implementations under the MIT license.

\begin{figure*}[ht]
    \centering
    \includegraphics[width=\textwidth]{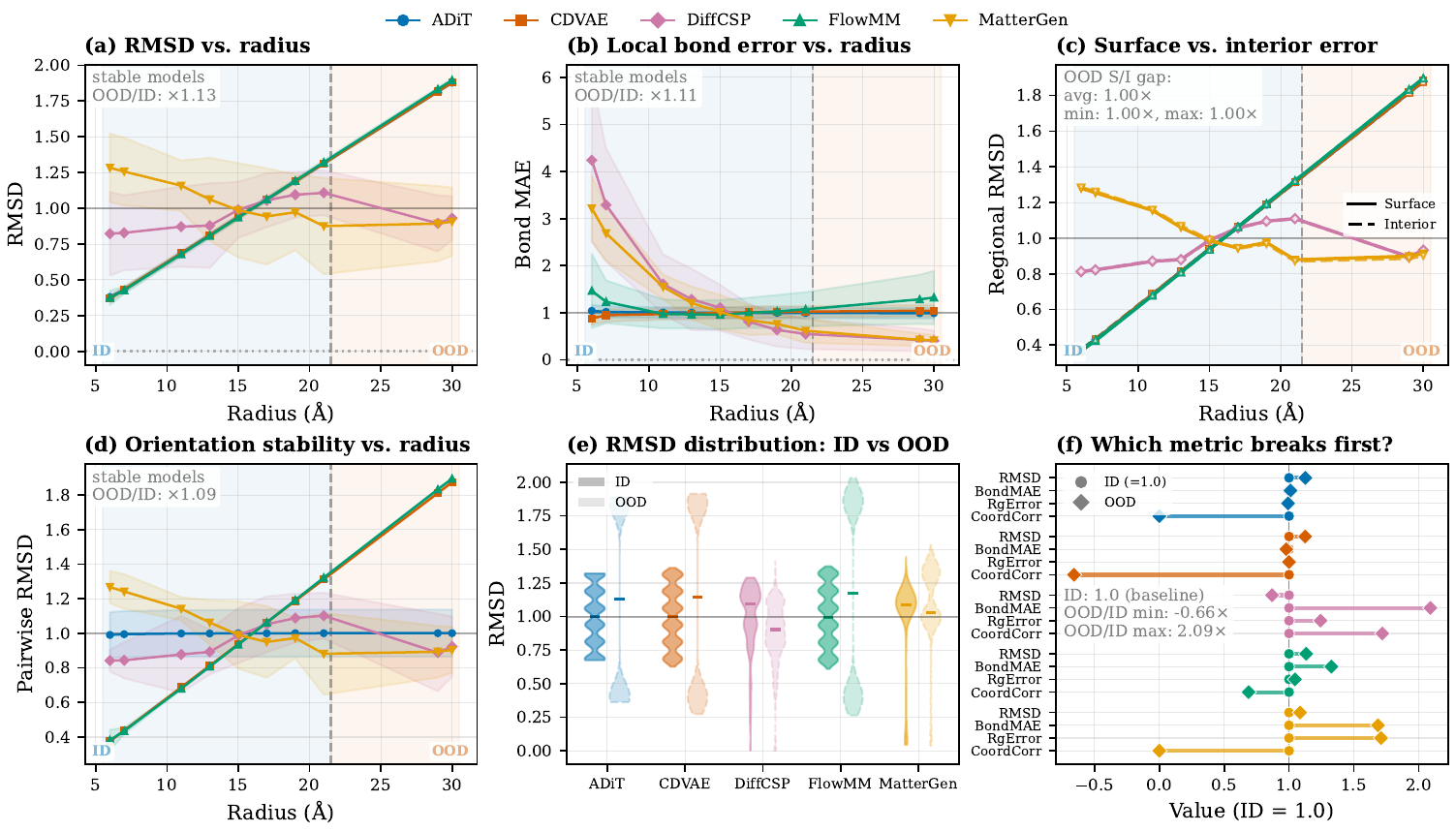}
    \Description{Frontier analysis}
    \caption{\textbf{Extrapolation frontier across multiple dimensions of structural fidelity.} \textbf{(a)} Global RMSD increases beyond the in-distribution boundary, revealing differing degrees of OOD degradation across models. \textbf{(b)} Local bond geometry errors diverge more strongly, showing that extrapolation behavior depends on the evaluation metric. \textbf{(c)} Surface atoms consistently exhibit higher errors than interior atoms, with similar trends from ID to OOD regimes, indicating uniform degradation. \textbf{(d)} Orientation consistency generally degrades alongside positional accuracy, though stability varies across architectures. \textbf{(e)} Distribution shifts in OOD samples show broader error tails compared to ID. \textbf{(f)} Multi-metric comparison highlights architecture-specific failure modes, with different models degrading along different structural dimensions. Shaded regions denote ID and OOD radii, and error bands indicate $\pm 1$ standard deviation across materials and seeds.}
    \label{fig:hero}
\end{figure*}

\section{Experiments}\label{sec:exp}

We evaluate five generative models (CDVAE \cite{cdvae}, DiffCSP \cite{jiao2023crystal}, FlowMM \cite{miller2024flowmm}, MatterGen \cite{zeni2023mattergen}, and ADiT \cite{joshi2025all}) on the unit-cell$\,\to\,$nanoparticle task (Eq.~\eqref{eq:task}). Each model receives the unit cell, target radius, and atom count as conditioning and generates exactly $N(R)$ atoms with the specified composition, ensuring one-to-one correspondence with the ground truth (Section~\ref{subsec:task}). The species sequence and atom ordering are communicated to each model as an explicit input tensor that fixes the identity and position index of every atom; during generation, models produce coordinates for each index in this fixed sequence, so the output is order-aligned with the ground truth by construction rather than by post-hoc matching. All metrics are computed per radius, averaged over all orientations at that radius. We report both \emph{raw} values (in~\AA) and \emph{normalized} values where each model's ID mean is set to~1.0 so that OOD values directly express relative degradation; we stress that normalized ratios are interpretable only for models whose raw ID performance is structurally meaningful (see Section~\ref{subsec:exp_profiles} for absolute-scale context). All experiments are repeated over three random seeds; we report means and $\pm 1$ standard deviation throughout. Seed-to-seed variability in the RMSD degradation ratio is small for the three well-behaved models (ADiT, CDVAE, FlowMM; std $< 0.02$); the divergent models show larger seed spread, itself a failure signal.

\paragraph*{Training protocol.}
All five architectures are trained under a unified protocol to ensure a fair comparison (Table~\ref{tab:training}). All models share the same optimizer, batch size, gradient clipping, learning rate scheduler, and training until plateau. Radius conditioning is provided as an explicit scalar input alongside the unit cell. Architecture-specific parameters (hidden dimensions, cutoff radii, diffusion/flow schedules) follow published specifications scaled to a unified budget of ${\sim}500$--$550$K parameters to isolate architectural differences from capacity effects. We acknowledge that this budget may disadvantage architectures designed for larger scales; Section~\ref{sec:limitation} discusses this trade-off, and Section~\ref{subsec:exp_capacity} probes it directly by rerunning the divergent models at their published parameter counts. At the largest radii ($R=30$~\AA, up to 11{,}298 atoms), peak GPU memory ranged from ${\sim}8$~GB (CDVAE) to ${\sim}18$~GB (ADiT); all models used the same neighbor-cutoff-based graph construction (Table~\ref{tab:training}) without additional sparsification.

\begin{table}[t]
\centering
\caption{Shared hyperparameters: Adam optimizer (lr $=10^{-4}$), batch size 2, gradient clip norm 1.0, ReduceLROnPlateau (factor 0.5, patience 5), and 3 seeds.}
\label{tab:training}
\resizebox{\columnwidth}{!}{%
\begin{tabular}{lccccc}
\toprule
& \textbf{ADiT} & \textbf{CDVAE} & \textbf{DiffCSP} & \textbf{FlowMM} & \textbf{MatterGen} \\
\midrule
Hidden dim & 24 & 92 & 120 & 120 & 130 \\
Num layers & 2 & 2 & 2 & 2 & 2 \\
Cutoff (\AA) & 5.0 & 5.0 & 7.0 & 5.0 & 5.0 \\
Latent dim & 8 & 92 & --- & --- & --- \\
Diffusion steps & 1000 & --- & 1000 & --- & 100 \\
$\beta$ range & $[10^{-4}, 0.02]$ & --- & $[10^{-4}, 0.02]$ & --- & --- \\
$\sigma$ range & --- & $[0.01, 1.0]$ & --- & $[10^{-4}, \text{---}]$ & --- \\
Num Gaussians & --- & 50 & 50 & 50 & 50 \\
\bottomrule
\end{tabular}%
}
\end{table}

\subsection{Where Is the Frontier?}
\label{subsec:exp_profiles}

Figure~\ref{fig:hero}(a) plots normalized RMSD versus radius. ADiT, CDVAE, and FlowMM show tightly clustered degradation ratios of 1.13, 1.12, and 1.13, corresponding to a consistent ${\sim}13\%$ RMSD increase from ID to OOD. DiffCSP’s normalized OOD RMSD decreases to 0.87, but this apparent robustness is misleading: its raw ID RMSD exceeds 3{,}386~\AA\ (${\sim}290$--$470\times$ larger than ADiT at 7.15~\AA\ or FlowMM at 11.55~\AA), indicating structurally incoherent outputs at all scales, with MatterGen’s 5{,}905~\AA\ RMSD placing it in the same regime. Their normalized ratios therefore reflect relative change atop already failed baselines. At the shared ${\sim}500$K-parameter budget these two samplers are numerically divergent rather than merely inaccurate. We verified these magnitudes through multiple diagnostics: consistent \AA\ units, valid Kabsch solutions ($\det(\mathbf{R}^\star)=+1$), broadly distributed per-atom errors rather than outliers, visual inspection showing globally disordered structures, and agreement from complementary metrics (DiffCSP: BondMAE $2.09\times$, RgError $1.24\times$; MatterGen: BondMAE $1.69\times$, RgError $1.71\times$). These models were originally designed for periodic crystals at smaller atom counts, suggesting a combination of task mismatch and parameter-budget constraints rather than a pipeline artifact. Figure~\ref{fig:hero}(b) shows a contrasting frontier under local bond fidelity: BondMAE degradation diverges substantially (DiffCSP $2.09\times$, MatterGen $1.69\times$, FlowMM $1.33\times$) while ADiT (1.01) and CDVAE (0.98) remain near unity, demonstrating that \textbf{the extrapolation frontier is metric-dependent}. Models appearing robust under global RMSD may simultaneously lose bond-length distributions required for chemical validity.

\subsection{Where Do Failures Originate?}
\label{subsec:exp_surface}

A natural hypothesis is that extrapolation failures concentrate on under-coordinated surface atoms. Figure~\ref{fig:hero}(c) partitions atoms into surface (outer 25\%) and interior (inner 25\%) shells defined on the \emph{ground-truth} structure, with per-atom errors computed from the aligned prediction. Across all five models, the surface--interior gap ratio remains remarkably stable from ID to OOD, with changes bounded by $\pm 0.0227$ at the default $25\%$ shell fraction ($\pm 0.0475$ across fractions from $15\%$ to $30\%$; Appendix~\ref{app:sensitivity}). This small shift rules out boundary-driven collapse: degradation propagates uniformly through the structure, indicating that the deficit lies in modeling how bulk structure extends with size rather than in surface geometry.

\subsection{Does Orientation Stability Transfer Across Scale?}
\label{subsec:exp_orientation}

By evaluating each structure under multiple input rotations, RADII isolates orientation sensitivity as a dimension independent of positional accuracy. Figure~\ref{fig:hero}(d) tracks orientation consistency across radii. CDVAE and FlowMM show orientation-consistency degradation ratios of 1.14 and 1.21; FlowMM's orientation degrades more steeply than its RMSD ($1.13\times$), so orientation and positional accuracy are not interchangeable frontiers. ADiT keeps its absolute orientation error negligibly small at every scale (${\sim}0.003$--$0.005$~\AA\ from ID to OOD), behaving as an effectively orientation-equivariant model even where positional accuracy degrades. Orientation consistency thus constitutes a \emph{secondary frontier}, architecturally decoupled from the primary quality frontier.

\subsection{How Are Errors Distributed?}
\label{subsec:exp_violin}

Beyond mean degradation, distributional analysis exposes worst-case risk. Figure~\ref{fig:hero}(e) shows full RMSD distributions under ID and OOD conditions. ADiT, CDVAE, and FlowMM exhibit matching tail ratios (Q95$_{\mathrm{OOD}}$/Q95$_{\mathrm{ID}}$) of 1.43, meaning worst-case OOD errors are 43\% worse than worst ID cases. FlowMM shows the largest median shift ($+0.176$ in normalized RMSD) and highest normalized OOD maximum ($2.03$), indicating occasional catastrophic failures at extrapolation scales. DiffCSP's distribution narrows (tail ratio 0.91), consistent with collapse into a scale-insensitive failure mode where all outputs are uniformly poor.

\subsection{Which Metric Breaks First?}
\label{subsec:exp_sequence}

Figure~\ref{fig:hero}(f) compares ID$\,\to\,$OOD degradation across four normalized metrics simultaneously. The failure sequences are qualitatively distinct: ADiT fails primarily on global RMSD ($1.13\times$) while local chemistry is preserved; FlowMM breaks on BondMAE ($1.33\times$) before RMSD; MatterGen suffers simultaneous RgError ($1.71\times$) and BondMAE ($1.69\times$) collapse; and DiffCSP is dominated by a BondMAE collapse ($2.09\times$), with coordination correlation falling to ${\approx}0$.
No two architectures exhibit the same sequence of metric failures, demonstrating that the frontier is a multi-dimensional surface shaped by model family.

To provide practitioner-ready comparisons, Table~\ref{tab:frontier} instantiates the frontier radius $r^\star(m,\tau)$ for both RMSD and BondMAE: ADiT reaches the farthest under RMSD ($r^\star=30$~\AA\ at $\tau=15$~\AA), FlowMM is the only model with a finite BondMAE frontier, and DiffCSP and MatterGen have no finite frontier at any threshold (full threshold sweep in Appendix~\ref{app:results}). Material dependence reinforces this: RMSD degradation is nearly material-invariant (std $= 0.005$, range 1.118--1.139), whereas BondMAE spread widens dramatically (std $= 0.208$, range 0.962--1.701), with tetragonal oxides TiO$_2$ and SnO$_2$ hardest and cubic metals easiest, correlating with unit-cell complexity.

\begin{table}[t]
\centering
\caption{Frontier radius $r^\star(m,\tau)$ in \AA: the largest radius at which model $m$ holds the metric within threshold $\tau$; ``--'' means no radius qualifies. Full threshold sweep in Appendix~\ref{app:results}.}
\label{tab:frontier}
\begin{tabular}{lcccc}
\toprule
& \multicolumn{3}{c}{RMSD $r^\star$ (\AA)} & BondMAE $r^\star$ \\
\cmidrule(lr){2-4}\cmidrule(lr){5-5}
\textbf{Model} & $\tau{=}5$ & $\tau{=}10$ & $\tau{=}15$ & $\tau{=}0.6$ \\
\midrule
ADiT      & 11 & 21 & 30 & -- \\
CDVAE     & -- & 7  & 13 & -- \\
DiffCSP   & -- & -- & -- & -- \\
FlowMM    & 7  & 13 & 19 & 21 \\
MatterGen & -- & -- & -- & -- \\
\bottomrule
\end{tabular}
\end{table}

\subsection{Scaling Laws for Nanostructure Generation}
\label{subsec:exp_scaling}

Figure~\ref{fig:scaling} fits power-law relationships RMSD~$\sim N^{\alpha}$ on ID radii. ADiT ($\alpha=0.334$, $R^2=1.000$), CDVAE ($0.335$, $1.000$), and FlowMM ($0.342$, $1.000$) exhibit nearly identical exponents near $1/3$, indicating RMSD grows with the nanoparticle’s linear dimension ($N^{1/3}\propto R$). An exponent of $1/3$ is the \emph{expected} geometric baseline for a well-behaved model rather than a surprising finding; its diagnostic value lies in conformity versus deviation. Geometrically, this reflects spatially uniform positional error accumulating with structure size, corresponding to systematic scaling rather than abrupt failure. DiffCSP ($\alpha=0.142$, $R^2=0.928$) and MatterGen ($\alpha=-0.126$, $R^2=0.897$) deviate strongly, with MatterGen’s negative exponent arising from fitting noise under uniformly large errors. OOD residuals further distinguish predictable from unstable scaling: ADiT (0.001) and FlowMM (0.004) maintain near-zero residuals, meaning ID scaling accurately predicts OOD degradation, whereas DiffCSP (0.118) and MatterGen (0.050) diverge substantially. For models in the $\alpha\approx1/3$ regime, performance at unseen sizes can therefore be estimated from ID fits alone. Appendix~\ref{app:results} reports the full scaling-law fits: the three well-behaved models follow a consistent RMSD power law while the divergent architectures exhibit unstable exponents. Seed-to-seed variability of $\alpha$ is below $0.002$ for the conforming models, indicating the scaling law is a stable architectural property for them, whereas DiffCSP and MatterGen are seed-unstable (std $> 0.09$), itself a failure signal.

\begin{figure}[t]
    \centering
    \includegraphics[width=0.95\columnwidth]{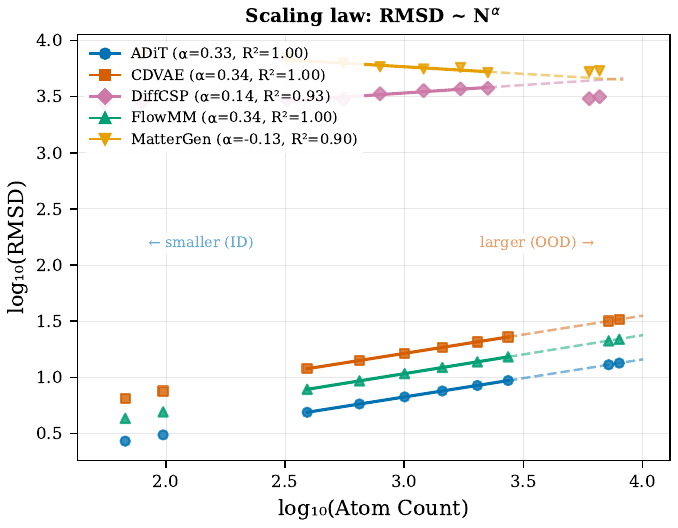}
    \Description{Scaling law}
    \caption{\textbf{Power-law scaling relationships quantify the extrapolation frontier.}
A log–log plot of RMSD versus atom count shows approximate power-law behavior, $\mathrm{RMSD} \sim N^\alpha$. Solid lines denote fits on in-distribution data, while dashed lines extrapolate into out-of-distribution regimes. Models with consistent scaling exhibit predictable extrapolation, whereas deviations indicate irregular or less predictable behavior beyond the training regime.}
    \label{fig:scaling}
\end{figure}

%─────────────────────────────────────────────────────────────────────────────
\subsection{Does Model Capacity Close the Frontier?}
\label{subsec:exp_capacity}

The shared ${\sim}500$K-parameter budget isolates architecture from capacity but may disadvantage models designed for larger scales. To separate the two effects, we retrained the two models that diverged at this budget (MatterGen and DiffCSP) at approximately their published parameter counts (MatterGen at 47M; DiffCSP at 12.4M), under the faster training settings detailed in Appendix~\ref{app:capacity}. The two outcomes differ. MatterGen at 47M recovers stable, bounded sampling (ID RMSD 17.63~\AA, OOD RMSD 19.81~\AA) where the 500K run had been numerically divergent; its shared-budget failure was therefore a capacity artifact. Yet the extrapolation frontier persists even at 47M: per-radius RMSD grows monotonically with $R$ and is largest at the farthest OOD radii. DiffCSP at 12.4M, by contrast, still diverges: reverse-diffusion sampling saturates the numerical safety bound despite a normally converging training loss, so capacity is not its bottleneck. Together these runs decompose the failure modes: capacity governs whether sampling is stable at all, whereas the size-extrapolation frontier is architectural. Finally, as a benchmark rather than a model, RADII's analog of an architecture ablation is the sensitivity analysis of its own design choices (Appendix~\ref{app:sensitivity}), not an ablation of model internals.

\section{Limitations and Ethical Considerations}\label{sec:limitation}
RADII evaluates generation on idealized nanoparticle geometries constructed by spherical truncation of bulk lattices without surface reconstruction, passivation, or finite-temperature relaxation, so performance reflects fidelity to ideal structures rather than experimentally realized surfaces, and both failure modes and success cases may change after structural relaxation. The current formulation conditions on atom count, species sequence, and ordering, isolating geometric extrapolation but narrowing generality; a natural extension is an unconditioned track evaluated with assignment-free metrics. The unified ${\sim}500$--$550$K parameter budget enables controlled architectural comparison but may disadvantage models designed for larger scales or different inductive biases, and several evaluated methods were originally developed for periodic crystals rather than large finite clusters. Metric design also introduces limitations: BondMAE$_k$ compares globally sorted neighbor distances and may conflate distinct local environments, while assignment-free measures such as Chamfer distance, earth mover's distance, or RDF divergences could provide complementary correspondence-independent validation. Finally, the ten selected materials span multiple bonding types and symmetries but do not exhaust inorganic chemistry, and the 6--30~\AA radius range probes only the small/medium-nanoparticle regime, leaving open whether the observed scaling behavior persists at larger sizes.

\paragraph*{Ethical considerations.}RADII comprises only computational structural data derived from published crystallographic references; it involves no human subjects, personal data, or dual-use concerns, and all code and benchmark data are released under the MIT license. Because the benchmark targets are idealized crystal-derived references, results should not be over-interpreted as direct validation for real-world nanomaterial deployment without additional physical relaxation or experimental verification.

\section{Conclusion}\label{sec:conc}

We introduced RADII, a radius-resolved benchmark that maps the extrapolation frontier of crystalline generative models across 25 size configurations, ten materials, and five architectures, revealing that well-behaved models degrade by ${\sim}13\%$ in global positional error beyond training radii while divergent models fail in absolute structural fidelity and local bond fidelity ranges from negligible degradation to more than 2× error growth, that no two architectures share the same failure sequence, that surface and interior errors grow in lockstep rather than from boundary effects, and that well-behaved models obey a power-law exponent $\alpha \approx 1/3$ whose in-distribution fit accurately predicts out-of-distribution error. These findings establish output scale as a first-class evaluation axis and reframe the extrapolation frontier as a diagnosable, forecastable quantity, though we recognize that the current benchmark evaluates a geometry-conditioned subtask on idealized references under a constrained parameter budget and that frontier locations may shift under more realistic conditions or larger model capacities. A first probe of the capacity question, rerunning the two divergent models at their published parameter counts, shows that greater capacity can stabilize sampling, as observed for MatterGen, but does not by itself close the size-extrapolation frontier; DiffCSP remains unstable even at published scale. Future work will incorporate DFT-relaxed and Wulff-shaped references, introduce an unconditioned evaluation track with assignment-free metrics, evaluate architectures at their intended scales alongside controlled-budget comparisons, explore size-conditioned training strategies to push the frontier outward, and investigate whether the scaling law and frontier-diagnostic framework transfer to proteins, biomolecular assemblies, and amorphous solids.
\section*{GenAI Disclosure}
Generative AI tools were used during the preparation of this manuscript to assist with editing prose, refining mathematical notation, and debugging code. All scientific contributions, including benchmark design, experimental methodology, data generation, model evaluation, and interpretation of results, were conceived and executed entirely by the authors. All AI-generated text was critically reviewed, verified, and revised by the authors, who take full responsibility for the content of this work.

\bibliographystyle{ACM-Reference-Format}
\bibliography{sample-base}

\appendix

\section{Dataset Details}\label{app:dataset}

RADII is constructed from ten primitive unit cells taken from published,
experimentally reported crystal structures. Table~\ref{tab:app_splits} summarizes the
leakage-free split, and Table~\ref{tab:app_atomcounts} reports the atom count of every
(material, radius) pair on the ID and OOD test radii. Atom count scales as $R^3$, so
the OOD radii correspond to structures well outside the training envelope: $33$ atoms
is $59\%$ smaller than the smallest training structure ($81$ atoms) and $11{,}298$
atoms is $24\%$ larger than the largest ($9{,}148$ atoms).

\begin{table}[H]
\centering
\footnotesize
\caption{Leakage-free radius split. Atom counts are over all materials.}
\label{tab:app_splits}
\begin{tabular}{lccc}
\toprule
\textbf{Split} & \textbf{Radii (\AA)} & \textbf{Structures} & \textbf{Atom range} \\
\midrule
Train & 15: $8$--$10,12,14,16,18,20,22$--$28$ & 48{,}000 & 81--9{,}148 \\
ID    & 6: $11,13,15,17,19,21$ & 13{,}500 & 203--3{,}858 \\
OOD   & 4: $6,7,29,30$          & 13{,}480 & 33--11{,}298 \\
\bottomrule
\end{tabular}
\end{table}

\begin{table}[H]
\centering
\footnotesize
\caption{Atom count per (material, radius) on the ID and OOD test radii.}
\label{tab:app_atomcounts}
\resizebox{\columnwidth}{!}{%
\begin{tabular}{lcccccccccc c}
\toprule
\textbf{Radius (\AA)} & \textbf{Ag} & \textbf{Au} & \textbf{CH$_3$NH$_3$PbI$_3$} & \textbf{Fe$_2$O$_3$} & \textbf{MoS$_2$} & \textbf{PbS} & \textbf{SnO$_2$} & \textbf{SrTiO$_3$} & \textbf{TiO$_2$} & \textbf{ZnO} & \textbf{Split} \\
\midrule
6  & 55   & 55   & 44   & 94    & 60   & 33   & 83   & 87   & 84    & 78   & OOD \\
7  & 79   & 79   & 63   & 136   & 84   & 57   & 105  & 119  & 132   & 117  & OOD \\
11 & 321  & 321  & 274  & 564   & 312  & 203  & 465  & 451  & 518   & 480  & ID  \\
13 & 555  & 555  & 402  & 914   & 522  & 365  & 777  & 781  & 862   & 747  & ID  \\
15 & 791  & 887  & 702  & 1400  & 792  & 515  & 1205 & 1227 & 1324  & 1185 & ID  \\
17 & 1205 & 1205 & 998  & 2032  & 1164 & 751  & 1717 & 1709 & 1908  & 1749 & ID  \\
19 & 1721 & 1721 & 1370 & 2862  & 1626 & 1045 & 2415 & 2387 & 2698  & 2415 & ID  \\
21 & 2243 & 2315 & 1886 & 3858  & 2178 & 1503 & 3249 & 3217 & 3614  & 3228 & ID  \\
29 & 5979 & 6051 & 4966 & 10180 & 5784 & 3887 & 8589 & 8537 & 9580  & 8565 & OOD \\
30 & 6603 & 6699 & 5486 & 11298 & 6402 & 4385 & 9457 & 9503 & 10558 & 9477 & OOD \\
\bottomrule
\end{tabular}%
}
\end{table}

\section{Full Quantitative Results}\label{app:results}

Table~\ref{tab:app_metrics} reports every metric on the ID and OOD splits (mean
$\pm$ standard deviation over three seeds), Table~\ref{tab:app_degrad} the ID$\to$OOD
degradation ratios, Table~\ref{tab:app_frontier} the full frontier-radius sweep, and
Table~\ref{tab:app_scaling} the power-law scaling fits. For DiffCSP and MatterGen the
$\sim$500K-budget runs are numerically divergent (Appendix~\ref{app:capacity}), so
their absolute magnitudes report the scale of divergent outputs rather than meaningful
structural error.

\begin{table}[H]
\centering
\footnotesize
\caption{Per-model metrics on ID and OOD splits (mean $\pm$ std, three seeds).
RMSD, BondMAE, RgError in \AA; CoordCorr and SurfIntRatio dimensionless.}
\label{tab:app_metrics}
\resizebox{\columnwidth}{!}{%
\begin{tabular}{lcccccccccc}
\toprule
& \multicolumn{2}{c}{\textbf{RMSD}} & \multicolumn{2}{c}{\textbf{BondMAE}} & \multicolumn{2}{c}{\textbf{RgError}} & \multicolumn{2}{c}{\textbf{CoordCorr}} & \multicolumn{2}{c}{\textbf{SurfIntRatio}} \\
\cmidrule(lr){2-3}\cmidrule(lr){4-5}\cmidrule(lr){6-7}\cmidrule(lr){8-9}\cmidrule(lr){10-11}
\textbf{Model} & ID & OOD & ID & OOD & ID & OOD & ID & OOD & ID & OOD \\
\midrule
ADiT      & $7.15{\pm}0.00$ & $8.05{\pm}0.00$ & $2.55{\pm}0.01$ & $2.58{\pm}0.01$ & $0.98{\pm}0.00$ & $0.98{\pm}0.00$ & $0.00$ & $0.00$ & $1.96{\pm}0.00$ & $1.96{\pm}0.00$ \\
CDVAE     & $17.46{\pm}0.00$ & $19.64{\pm}0.00$ & $3.39{\pm}0.00$ & $3.32{\pm}0.00$ & $1.24{\pm}0.00$ & $1.24{\pm}0.00$ & $0.00$ & $-0.00$ & $1.09{\pm}0.00$ & $1.09{\pm}0.00$ \\
DiffCSP   & $3386.5{\pm}134.2$ & $2944.6{\pm}430.0$ & $256.6{\pm}128.6$ & $536.4{\pm}125.6$ & $486.5{\pm}12.3$ & $604.6{\pm}66.7$ & $0.00$ & $0.00$ & $1.01{\pm}0.02$ & $1.03{\pm}0.03$ \\
FlowMM    & $11.55{\pm}0.02$ & $13.08{\pm}0.08$ & $0.55{\pm}0.04$ & $0.73{\pm}0.10$ & $0.28{\pm}0.00$ & $0.29{\pm}0.03$ & $0.02{\pm}0.00$ & $0.02{\pm}0.02$ & $1.24{\pm}0.00$ & $1.24{\pm}0.01$ \\
MatterGen & $5904.6{\pm}1175.8$ & $6408.6{\pm}873.4$ & $1751.6{\pm}328.9$ & $2953.3{\pm}313.7$ & $884.0{\pm}162.9$ & $1511.5{\pm}164.3$ & $0.00$ & $0.00$ & $1.01{\pm}0.00$ & $1.01{\pm}0.00$ \\
\bottomrule
\end{tabular}%
}
\end{table}

\begin{table}[H]
\centering
\footnotesize
\caption{ID$\to$OOD degradation ratios (OOD mean / ID mean). CoordCorr is omitted:
its ID and OOD values are $\approx 0$ for several models, making the ratio
uninformative (see absolute values in Table~\ref{tab:app_metrics}).}
\label{tab:app_degrad}
\begin{tabular}{lcccc}
\toprule
\textbf{Model} & \textbf{RMSD} & \textbf{BondMAE} & \textbf{RgError} & \textbf{SurfIntRatio} \\
\midrule
ADiT      & 1.13 & 1.01 & 0.99 & 1.00 \\
CDVAE     & 1.12 & 0.98 & 1.00 & 1.00 \\
DiffCSP   & 0.87 & 2.09 & 1.24 & 1.02 \\
FlowMM    & 1.13 & 1.33 & 1.05 & 1.00 \\
MatterGen & 1.09 & 1.69 & 1.71 & 1.00 \\
\bottomrule
\end{tabular}
\end{table}

\begin{table}[H]
\centering
\footnotesize
\caption{Frontier radius $r^\star(m,\tau)$ in \AA{} across thresholds. ``--'' = no
radius satisfies the threshold. ADiT extends farthest under RMSD; FlowMM is the only
model with a finite BondMAE frontier; DiffCSP and MatterGen have none.}
\label{tab:app_frontier}
\resizebox{\columnwidth}{!}{%
\begin{tabular}{lcccccc c cccccc}
\toprule
& \multicolumn{6}{c}{\textbf{RMSD threshold $\tau$ (\AA)}} & & \multicolumn{6}{c}{\textbf{BondMAE threshold $\tau$ (\AA)}} \\
\cmidrule(lr){2-7}\cmidrule(lr){9-14}
\textbf{Model} & 1.0 & 2.0 & 5.0 & 8.0 & 10.0 & 15.0 & & 0.3 & 0.5 & 0.6 & 0.7 & 0.8 & 1.0 \\
\midrule
ADiT      & -- & -- & 11 & 17 & 21 & 30 & & -- & -- & -- & -- & -- & -- \\
CDVAE     & -- & -- & -- & 7  & 7  & 13 & & -- & -- & -- & -- & -- & -- \\
DiffCSP   & -- & -- & -- & -- & -- & -- & & -- & -- & -- & -- & -- & -- \\
FlowMM    & -- & -- & 7  & 11 & 13 & 19 & & -- & -- & 21 & 21 & 30 & 30 \\
MatterGen & -- & -- & -- & -- & -- & -- & & -- & -- & -- & -- & -- & -- \\
\bottomrule
\end{tabular}%
}
\end{table}

\begin{table}[H]
\centering
\footnotesize
\caption{Power-law fit $\log_{10}\mathrm{RMSD} = \alpha\log_{10}N + c$ on ID radii.
OOD residual is the mean absolute prediction error when the ID fit is extrapolated to
OOD radii.}
\label{tab:app_scaling}
\begin{tabular}{lcccc}
\toprule
\textbf{Model} & \textbf{$\alpha$} & \textbf{$R^2$} & \textbf{Intercept} & \textbf{OOD residual} \\
\midrule
ADiT      & $0.334$  & $1.000$ & $-0.176$ & $0.001$ \\
CDVAE     & $0.335$  & $1.000$ & $0.209$  & $0.004$ \\
DiffCSP   & $0.142$  & $0.928$ & $3.107$  & $0.118$ \\
FlowMM    & $0.342$  & $1.000$ & $0.007$  & $0.004$ \\
MatterGen & $-0.126$ & $0.897$ & $4.145$  & $0.050$ \\
\bottomrule
\end{tabular}
\end{table}

\section{Metric Sensitivity Analysis}\label{app:sensitivity}

We sweep the two free parameters of the structural metrics. The CoordCorr cutoff
$r_c$ (Table~\ref{tab:app_coord}) yields cross-architecture rankings that are
identical to the default $r_c = 3.0$~\AA{} for all $r_c \in [3.0, 4.5]$~\AA{}
(Spearman $\rho = 1.0$) and nearly identical at $5.0$~\AA{} ($\rho = 0.9$); only the
very short cutoffs $2.0$--$2.5$~\AA{}, which capture partial first shells, reorder
models. The SurfIntRatio shell fraction (Table~\ref{tab:app_surfint}) leaves the
ID$\to$OOD shift small across $15$--$30\%$: at the default $25\%$ every model satisfies
$|\Delta| \le 0.0227$, and the maximum over all fractions is $0.0475$. Finally, a
per-atom BondMAE that compares each atom's neighbor list individually
(Table~\ref{tab:app_peratom}) reproduces the global-metric ranking exactly
(Spearman $\rho = 1.0$). Per-material breakdowns of all three sweeps are consistent
with these aggregates and are released with the benchmark repository.

\begin{table}[H]
\centering
\footnotesize
\caption{CoordCorr cutoff sweep: Spearman $\rho$ of the cross-architecture ranking
against the default $r_c = 3.0$~\AA.}
\label{tab:app_coord}
\begin{tabular}{lccccccc}
\toprule
\textbf{$r_c$ (\AA)} & 2.0 & 2.5 & 3.0 & 3.5 & 4.0 & 4.5 & 5.0 \\
\midrule
\textbf{Spearman $\rho$} & $-0.20$ & $-0.10$ & $1.00$ & $1.00$ & $1.00$ & $1.00$ & $0.90$ \\
\bottomrule
\end{tabular}
\end{table}

\begin{table}[H]
\centering
\footnotesize
\caption{SurfIntRatio ID$\to$OOD shift $\Delta = \mathrm{OOD}-\mathrm{ID}$ across
shell fractions. The default fraction is $25\%$.}
\label{tab:app_surfint}
\begin{tabular}{lccccc}
\toprule
\textbf{Fraction} & \textbf{ADiT} & \textbf{CDVAE} & \textbf{DiffCSP} & \textbf{FlowMM} & \textbf{MatterGen} \\
\midrule
0.15 & $+0.0463$ & $+0.0001$ & $+0.0475$ & $-0.0019$ & $+0.0023$ \\
0.20 & $+0.0379$ & $-0.0012$ & $+0.0276$ & $-0.0017$ & $+0.0032$ \\
0.25 & $+0.0012$ & $-0.0025$ & $+0.0227$ & $-0.0023$ & $+0.0037$ \\
0.30 & $-0.0066$ & $-0.0025$ & $+0.0146$ & $-0.0033$ & $+0.0031$ \\
\bottomrule
\end{tabular}
\end{table}

\begin{table}[H]
\centering
\footnotesize
\caption{Per-atom vs.\ global BondMAE (\AA). The per-atom variant compares each
atom's neighbor list individually; model ranking is unchanged (Spearman $\rho = 1.0$).}
\label{tab:app_peratom}
\begin{tabular}{lcccc}
\toprule
& \multicolumn{2}{c}{\textbf{Per-atom}} & \multicolumn{2}{c}{\textbf{Global}} \\
\cmidrule(lr){2-3}\cmidrule(lr){4-5}
\textbf{Model} & ID & OOD & ID & OOD \\
\midrule
ADiT      & 2.55   & 2.58   & 2.55   & 2.58   \\
CDVAE     & 3.40   & 3.33   & 3.39   & 3.32   \\
DiffCSP   & 256.60 & 536.35 & 256.60 & 536.35 \\
FlowMM    & 0.74   & 0.91   & 0.55   & 0.73   \\
MatterGen & 1751.58 & 2953.33 & 1751.58 & 2953.33 \\
\bottomrule
\end{tabular}
\end{table}

\section{Published-Scale Capacity Experiment}\label{app:capacity}

To separate capacity limits from architectural ones, the two models that diverged at
the shared $\sim$500K budget were retrained near their published parameter counts:
MatterGen at 47M parameters (hidden dimension 1100, 5 layers) and DiffCSP at 12.4M
parameters (hidden dimension 465, 6 layers). To keep these larger runs tractable they
used faster training settings than the main benchmark: half the training data
($\mathrm{loaded\_frac}=0.5$) and 50 epochs, with the evaluation pipeline otherwise
unchanged; the numbers here are therefore an indicative capacity probe rather than a
capacity-matched re-benchmark.

Table~\ref{tab:app_capacity} contrasts the two outcomes. MatterGen at 47M recovers
stable, bounded sampling: per-radius RMSD grows monotonically with $R$ (from
$12.05$~\AA{} at $R{=}11$ to $23.19$~\AA{} at $R{=}21$) and peaks at the farthest OOD
radii ($33.19$~\AA{} at $R{=}30$), so the extrapolation frontier remains visible even
at published scale. Per-material RMSD is uniform (ID $17.56$--$17.66$~\AA,
OOD $19.74$--$19.84$~\AA), and predictions are bounded within $\pm113$~\AA{} (ID) and
$\pm169$~\AA{} (OOD). DiffCSP at 12.4M, in contrast, still diverges: reverse-diffusion
sampling saturates the $\pm10{,}000$~\AA{} numerical safety bound for all ten materials
across both splits, even though training loss converges normally (final validation
$\approx 1.10$). Its per-radius RMSD shows no monotonic size trend, with error largest
at the smallest radii, consistent with lattice-diffusion runaway rather than a
size-generalization frontier. Capacity therefore governs whether sampling is stable at
all (MatterGen) but does not, by itself, resolve the failure for every architecture
(DiffCSP).

\begin{table}[H]
\centering
\footnotesize
\caption{Per-radius RMSD (\AA) at published scale. MatterGen-47M is stable and grows
monotonically with $R$; DiffCSP-12.4M remains bound-clipped and divergent, so its
values report the magnitude of divergent outputs, not structural error.}
\label{tab:app_capacity}
\begin{tabular}{cccc}
\toprule
\textbf{Radius (\AA)} & \textbf{Split} & \textbf{MatterGen-47M} & \textbf{DiffCSP-12.4M} \\
\midrule
6  & OOD & 6.44  & 2{,}314.9 \\
7  & OOD & 7.54  & 2{,}236.1 \\
11 & ID  & 12.05 & 1{,}688.9 \\
13 & ID  & 14.28 & 1{,}708.0 \\
15 & ID  & 16.53 & 1{,}834.3 \\
17 & ID  & 18.74 & 2{,}060.9 \\
19 & ID  & 20.97 & 1{,}954.9 \\
21 & ID  & 23.19 & 1{,}781.1 \\
29 & OOD & 32.07 & 1{,}669.2 \\
30 & OOD & 33.19 & 1{,}655.1 \\
\midrule
\multicolumn{2}{c}{\textbf{Overall} (ID / OOD)} & 17.63 / 19.81 & 1{,}838.0 / 1{,}968.8 \\
\bottomrule
\end{tabular}
\end{table}

\end{document}